\documentclass[pra,onecolumn,groupedaddress,superscriptaddress,amsmath,amssymb,longbibliography]{revtex4-2}
\usepackage{amsmath,amssymb,bm,amsthm}
\usepackage{graphicx,caption,subcaption}
\usepackage{enumerate}  
\usepackage[colorlinks=true,citecolor=blue,linkcolor=blue,urlcolor=blue]{hyperref}

\newtheorem{proposition}{Proposition}
\newtheorem{lemma}{Lemma}
\newtheorem{fact}{Fact}
\newtheorem{remark}{Remark}

\newcommand{\calE}{\mathcal E}
\newcommand{\id}{\mathcal I}

\newcommand{\ket}[1]{\lvert #1\rangle}
\newcommand{\bra}[1]{\langle #1\rvert}
\newcommand{\ketbra}[2]{\lvert #1\rangle\langle #2\rvert}
\newcommand{\abs}[1]{\left\lvert #1\right\rvert}
\newcommand{\norm}[1]{\left\lVert #1\right\rVert}
\newcommand{\basisB}{\mathcal B}

\begin{document}
\title{Entanglement-Assisted Timing Optimization for Discriminating
Amplitude-Damping Dynamics}

\author{Massimiliano F. Sacchi}
\affiliation{CNR-Istituto di Fotonica e Nanotecnologie, Piazza Leonardo da Vinci 32, I-20133, Milano, Italy} 
\affiliation{Dipartimento di Fisica, Universit\`a di Pavia, Via A. Bassi 6, I-27100, Pavia, Italy}
\author{Milajiguli Rexiti}
\author{Stefano Mancini}
\affiliation{School of Science and Technology, University of Camerino,
Via Madonna delle Carceri 9, Camerino I-62032, Italy}
\affiliation{INFN--Sezione Perugia, Via A. Pascoli, Perugia I-06123, Italy}


\begin{abstract}
We analyze minimum-error discrimination of two qubit dynamical processes
generated by phase-covariant amplitude-damping Lindbladians in a single-use
scenario.  The optimization involves both the input probe, possibly entangled
with an isolated ancilla, and the interrogation time.  For unassisted probes we
obtain a closed expression for the optimal trace-norm distinguishability at
fixed time, with distinct interior and boundary branches.  For
entanglement-assisted probes, the common phase covariance of the two channels
reduces the diamond-norm optimization to a one-parameter Schmidt family.  The
resulting formula gives transparent sufficient conditions for fixed-time
entanglement advantage and separates this local advantage from advantage after
global optimization over time.  We exhibit examples in which the best
unassisted strategy is approached only asymptotically, whereas an entangled
probe achieves a strictly smaller error probability at finite time.
\end{abstract}

\maketitle

\section{Introduction}
\label{sec:introduction}
Discrimination tasks provide one of the most direct operational ways of
quantifying the distinguishability of quantum states, channels, and dynamical
processes.  In the minimum-error setting, the optimal discrimination of two
quantum states is governed by the Helstrom trace-norm formula
\cite{Helstrom76,BaeKwek15}.  For quantum channels the analogous problem is
richer: one must optimize over the input probe, and the use of an ancillary
system left untouched by the channel leads to the completely bounded trace norm,
or diamond norm \cite{Sacchi05a,Watrous18}.  It is now well established that
side entanglement can strictly enhance channel discrimination, even for strongly
noisy channels and, in particular, for entanglement-breaking maps
\cite{Sacchi05c,PianiWatrous09,Oskouei23}.  Other discrimination paradigms,
including minimax, unambiguous, perfect, adaptive, and sequential channel
discrimination, further show that the operational separation between strategies
can depend sensitively on the resources available to the experimenter
\cite{dsk,Wang06,Duan09,Chiribella08,Harrow10}.

In many physical situations the alternatives to be discriminated are not static
channels but dynamical processes.  A probe interacts with an open system for a
chosen interrogation time, after which a measurement is performed.  The choice
of this time is itself part of the discrimination strategy: waiting longer may
increase the difference between stationary states, but it can also erase useful
coherences or make distinct transients indistinguishable.  The resulting
problem therefore combines two optimizations, one over input states and one over
time.  This time-dependent aspect is especially relevant for Markovian
open-system evolutions generated by Gorini--Kossakowski--Sudarshan--Lindblad
semigroups \cite{Gorini76,Lindblad76,BreuerPetruccione02,RivasHuelga12}.  It
also connects with recent work on the discrimination of dynamical processes, in
which optimal timing was shown to interact nontrivially with side entanglement
for Pauli semigroups \cite{Sacchi26}.  Closely related operational questions arise in quantum thermometry and bath tagging, where quantum probes discriminate alternative open-system evolutions induced by reservoirs with different temperatures, spectral structures, or statistical properties \cite{Jevtic15,Candeloro21,Farina19,Farina22}.

The present work studies this problem for phase-covariant amplitude-damping
semigroups. We use the term amplitude-damping semigroup in the generalized finite-temperature
sense, allowing both excitation and relaxation together with pure
dephasing. These qubit dynamics provide a simple but nontrivial class of generally
non-unital Markovian processes \cite{Filippov20,Siudzinska22}.  This non-unital
character is essential for the structure of the optimization.  For Pauli
channels, the discrimination problem is diagonal in the Pauli basis and
maximally entangled probes have a particularly simple status
\cite{Sacchi05b,Sacchi26}.  By contrast, amplitude-damping processes contain
both a longitudinal contraction and a translation of the Bloch ball.  As a
consequence, the optimal unassisted probe need not be an eigenstate of a Pauli
operator, and the optimal entangled input need not be maximally entangled.  The
amplitude-damping and dephasing cases have also been studied from complementary
perspectives, including channel-specific analyses and general bounds on process
discrimination \cite{Rexiti21,Nakahira21,Rexiti22,Cooney16}.

Our main results are as follows.  First, for unassisted probes we derive a
closed formula for the optimal trace-norm distinguishability at each time.  The
formula displays two possible branches: an interior branch, where a coherent
superposition is optimal, and a boundary branch, where an energy eigenstate is
optimal.  Second, for entanglement-assisted probes we exploit the common
$U(1)$ phase covariance of the two candidate channels.  Averaging the
diamond-norm semidefinite program over this symmetry shows that an optimal
assisted input can be chosen in a one-parameter Schmidt family.  This yields a
single-variable expression for the diamond-norm distinguishability and gives
transparent fixed-time criteria for when side entanglement is useful.  Third,
we optimize over the interrogation time and exhibit regimes in which the best
unassisted strategy is only asymptotically optimal, while an entangled probe
achieves a strictly smaller error probability at finite time.

The paper is organized as follows.  Section~\ref{sec:framework-model}
introduces the minimum-error discrimination setting, both with and without
side entanglement, and then specifies the phase-covariant amplitude-damping
semigroups considered throughout the paper.  Section~\ref{sec:unassisted}
derives the unassisted optimum.  Section~\ref{sec:assisted} proves the
Schmidt-family reduction for the assisted problem and gives sufficient
conditions for fixed-time entanglement advantage.  Section~\ref{sec:examples}
presents examples illustrating finite-time and asymptotic optima, threshold
behavior, and genuine non-unital entanglement-assisted advantages. The results
show that timing and entanglement are not independent resources in
dynamical process discrimination.  Entanglement can alter not only the
maximum distinguishability attainable from a single use of the
process, but also the time at which the optimum is effectively
reached.


\section{Discrimination framework and dynamical model}
\label{sec:framework-model}
We first recall the operational quantities used to compare the two
candidate dynamical processes and then specialize them to the
phase-covariant amplitude-damping semigroups studied in this work.
\subsection{Minimum-error discrimination with and without side entanglement}
\label{sec:general-framework}
Let the unknown dynamics be generated by one of two Lindbladians, labelled by
$i=0,1$, with equal prior probabilities.  At time $t$ the two alternatives
define channels $\calE_t^{(0)}$ and $\calE_t^{(1)}$.  Without side entanglement,
the minimum error probability is \cite{Sacchi05a}
\begin{equation}
 p_E(t)=\frac12-\frac14D_{\rm sep}(t),
 \label{eq:pE-sep-time}
\end{equation}
where
\begin{equation}
 D_{\rm sep}(t):=
 \max_{\rho}\norm{\calE_t^{(0)}(\rho)-\calE_t^{(1)}(\rho)}_1 .
\label{eq:Dsep-def}
\end{equation}
Here $\norm{\cdot}_1$ denotes the trace norm.  With an isolated ancilla, the
corresponding error probability is \cite{Sacchi05a}
\begin{equation}
 \widetilde p_E(t)=\frac12-\frac14D_{\rm ent}(t),
 \label{eq:pE-ent-time}
\end{equation}
where
\begin{equation}
 D_{\rm ent}(t):=
 \max_{\rho}\norm{(\id\otimes\calE_t^{(0)})(\rho)
 -(\id\otimes\calE_t^{(1)})(\rho)}_1 .
\label{eq:Dent-def-general}
\end{equation}
For qubit channels a two-dimensional ancilla is sufficient
\cite{Watrous18}.  The ultimate error probabilities are obtained by
optimizing over the interrogation time:
\begin{equation}
 p_E^*=\frac12-\frac14\sup_{t\ge0}D_{\rm sep}(t),
 \qquad
 \widetilde p_E^{\,*}=\frac12-\frac14\sup_{t\ge0}D_{\rm ent}(t).
\label{eq:ultimate-errors}
\end{equation}
Thus an assisted strategy is strictly better after time optimization precisely
when
\begin{equation}
 \sup_{t\ge0}D_{\rm ent}(t)>\sup_{t\ge0}D_{\rm sep}(t).
 \label{eq:global-assisted-criterion}
\end{equation}

\subsection{Phase-covariant amplitude-damping semigroups}
\label{sec:phase-covariant}
We consider the qubit master equation \cite{Gorini76,Lindblad76,Filippov20,Siudzinska22}
\begin{equation}
\begin{aligned}
\dot\rho(t)=&\;\gamma_+\left(\sigma_+\rho(t)\sigma_- -\frac12\{\sigma_-\sigma_+,\rho(t)\}\right)
+\gamma_-\left(\sigma_-\rho(t)\sigma_+ -\frac12\{\sigma_+\sigma_-,\rho(t)\}\right)\\
&+\gamma_z\left(\sigma_z\rho(t)\sigma_z-\rho(t)\right),
\end{aligned}
\label{eq:master}
\end{equation}
where all rates $\gamma_+,\gamma_-,\gamma_z$ are non-negative,
$\sigma_\pm=(\sigma_x\pm i\sigma_y)/2$, and
$\sigma_x,\sigma_y,\sigma_z$ are the Pauli operators.  The positive and
negative eigenstates of $\sigma_z$ are denoted by $\ket{0}$ and $\ket{1}$,
respectively.  The channel induced by Eq.~\eqref{eq:master},
\begin{equation}
 \rho(0)\longmapsto \rho(t)=\calE_t(\rho(0)),
 \label{eq:Emap}
\end{equation}
acts on the Bloch vector as
\begin{equation}
 (r_x,r_y,r_z)\longmapsto (\lambda r_x,\lambda r_y,\lambda_z r_z+t_z),
 \label{eq:bloch-action}
\end{equation}
with
\begin{equation}
 \lambda=e^{-\frac12(\Gamma+4\gamma_z)t},
 \qquad
 \lambda_z=e^{-\Gamma t},
 \qquad
 t_z=\chi(1-e^{-\Gamma t}),
 \label{eq:parameters}
\end{equation}
and
\begin{equation}
 \Gamma=\gamma_++\gamma_-,
 \qquad
 \chi=\frac{\gamma_+-\gamma_-}{\gamma_++\gamma_-}.
 \label{eq:Gamma-chi}
\end{equation}
If $\Gamma=0$, we use the limiting convention
\begin{equation}
 \lambda_z=1,
 \qquad
 t_z=0,
 \qquad
 \lambda=e^{-2\gamma_z t}.
 \label{eq:Gamma-zero-convention}
\end{equation}
For the two candidate dynamics define
\begin{equation}
 \alpha (t)=t_z^{(0)}-t_z^{(1)},
 \qquad
 \beta (t)=\lambda_z^{(0)}-\lambda_z^{(1)},
 \qquad
 \zeta (t)=\lambda^{(0)}-\lambda^{(1)}.
 \label{eq:abz}
\end{equation}
These are real functions of $t$ that 
represent the difference in translations, longitudinal
contractions, and transverse contractions of the Bloch ball,
respectively. For ease of notation, if not needed, we omit the
explicit time-dependence of $\alpha , \beta $ and $\zeta$. The difference map
\begin{equation}
 \Delta_t:=\calE_t^{(0)}-\calE_t^{(1)}
 \label{eq:Delta-def}
\end{equation}
is specified on the matrix units by
\begin{equation}
\begin{aligned}
\Delta_t(\ketbra{0}{0})&=\frac{\alpha+\beta}{2}\sigma_z,
&\qquad
\Delta_t(\ketbra{1}{1})&=\frac{\alpha-\beta}{2}\sigma_z,\\
\Delta_t(\ketbra{0}{1})&=\zeta\ketbra{0}{1},
&\qquad
\Delta_t(\ketbra{1}{0})&=\zeta\ketbra{1}{0}.
\end{aligned}
\label{eq:matrixunits}
\end{equation}
The map is covariant under phase rotations $U_\theta=e^{-i\theta\sigma_z/2}$:
\begin{equation}
 \Delta_t(U_\theta XU_\theta^\dagger)=U_\theta\Delta_t(X)U_\theta^\dagger .
 \label{eq:phase-covariance}
\end{equation}
In the bipartite calculations we use the ordered basis
\begin{equation}
 \basisB=\{\ket{0_A0_S},\ket{0_A1_S},\ket{1_A0_S},\ket{1_A1_S}\},
 \label{eq:bipartite-basis}
\end{equation}
where the first label refers to the ancilla and the second to the system.


\section{Unassisted discrimination}
\label{sec:unassisted}

By convexity of the trace norm, the optimization in Eq.~\eqref{eq:Dsep-def}
may be restricted to pure input states,
\begin{equation}
 \ket{\psi}=\sqrt{x}\ket{0}+e^{-i\phi}\sqrt{1-x}\ket{1},
 \qquad 0\le x\le1.
 \label{eq:pure-input}
\end{equation}
The trace norm of the output difference is independent of $\phi$ and equals
\begin{equation}
 D_{\rm sep}(x;t)=
 \sqrt{\left[\alpha+\beta(2x-1)\right]^2+4x(1-x)\zeta^2}.
 \label{eq:Dsep-x}
\end{equation}
Maximizing the quadratic expression under the square root gives
\begin{equation}
D_{\rm sep}(t)=
\begin{cases}
 \abs{\zeta}\sqrt{\displaystyle\frac{\alpha^2-\beta^2+\zeta^2}{\zeta^2-\beta^2}},
 & \zeta^2>\beta^2\ \text{and}\ \abs{\alpha\beta}\le \zeta^2-\beta^2,\\[2.2ex]
 \abs{\alpha}+\abs{\beta}, & \text{otherwise.}
\end{cases}
\label{eq:Dsep}
\end{equation}
In the first branch the maximizing population is
\begin{equation}
 x^*=\frac12\left(1+\frac{\alpha\beta}{\zeta^2-\beta^2}\right),
 \label{eq:xstar-unassisted}
\end{equation}
whereas in the second branch an eigenstate of $\sigma_z$ is optimal.

If both semigroups are genuinely relaxing, i.e. $\Gamma_0,\Gamma_1
>0$, notice that from Eq.~\eqref{eq:Dsep} for both branches one has
\begin{equation}
    \lim_{t\to\infty} D_{\rm sep}(t)=\abs{\chi_0-\chi_1}.
\end{equation}

\begin{fact}[Interior branch]
\label{cond1}
Let
\begin{equation}
I_{\rm int}:=
\left\{t>0:\zeta(t)^2>\beta(t)^2,
\ \abs{\alpha(t)\beta(t)}\le \zeta(t)^2-\beta(t)^2\right\}.
\label{eq:Iint}
\end{equation}
On every connected component of $I_{\rm int}$,
\begin{equation}
D_{\rm sep}(t)^2=F(t):=
\zeta(t)^2+\frac{\alpha(t)^2\zeta(t)^2}{\zeta(t)^2-\beta(t)^2}.
\label{eq:F-interior}
\end{equation}
Every finite maximizer lying in the interior of such a component satisfies
\begin{equation}
F'(t)=
2\zeta\zeta'
+\frac{2\alpha\alpha'\zeta^2}{\zeta^2-\beta^2}
+\frac{2\alpha^2\zeta\beta(\zeta\beta'-\beta\zeta')}{(\zeta^2-\beta^2)^2}=0.
\label{eq:Fprime}
\end{equation}
If a connected component is unbounded and $F'(t)<0$ for all sufficiently large
$t$, then any value of $F$ strictly larger than its asymptotic limit is attained
at a finite time.
\end{fact}

This fact follows from observing that, on $I_{\rm int}$ the first branch of Eq.~\eqref{eq:Dsep} applies, so
maximizing $D_{\rm sep}(t)$ is equivalent to maximizing the non-negative function
$F(t)$.  Differentiating Eq.~\eqref{eq:F-interior} gives Eq.~\eqref{eq:Fprime},
which is therefore necessary for an interior finite-time maximum.  On an
unbounded component with a decreasing tail, the supremum cannot be approached
only at infinity if some finite value exceeds the asymptotic limit; continuity
on compact intervals then gives a finite maximizer.

\begin{fact}[Boundary branch]
\label{cond2}
Assume that Eq.~\eqref{eq:Dsep} is on its boundary branch on a time interval,
so that $D_{\rm sep}(t)=\abs{\alpha(t)}+\abs{\beta(t)}$.  Let
\begin{equation}
 \Gamma_i=\gamma_+^{(i)}+\gamma_-^{(i)},
 \qquad
 \chi_i=\frac{\gamma_+^{(i)}-\gamma_-^{(i)}}{\Gamma_i},
 \qquad
 L=\chi_0-\chi_1,
 \label{eq:boundary-parameters}
\end{equation}
with $\Gamma_i>0$.  If $\alpha(t)\beta(t)\ge0$ on the interval, then
\begin{equation}
D_{\rm sep}(t)=\abs{h_-(t)},
\qquad
h_-(t)=L+\frac{2\gamma_-^{(0)}}{\Gamma_0}e^{-\Gamma_0t}
-\frac{2\gamma_-^{(1)}}{\Gamma_1}e^{-\Gamma_1t}.
\label{eq:hminus}
\end{equation}
If $0<\gamma_-^{(0)}<\gamma_-^{(1)}$ and $\Gamma_1>\Gamma_0$, the only
positive stationary point of $h_-$ is
\begin{equation}
t_-^*=\frac{\log\gamma_-^{(1)}-\log\gamma_-^{(0)}}{\Gamma_1-\Gamma_0},
\label{eq:tminus-star}
\end{equation}
and it is a local maximum of $h_-$.  Similarly, if $\alpha(t)\beta(t)<0$ on
the interval, then
\begin{equation}
D_{\rm sep}(t)=\abs{h_+(t)},
\qquad
h_+(t)=L-\frac{2\gamma_+^{(0)}}{\Gamma_0}e^{-\Gamma_0t}
+\frac{2\gamma_+^{(1)}}{\Gamma_1}e^{-\Gamma_1t}.
\label{eq:hplus}
\end{equation}
If $0<\gamma_+^{(0)}<\gamma_+^{(1)}$ and $\Gamma_1>\Gamma_0$, the only
positive stationary point of $h_+$ is
\begin{equation}
t_+^*=\frac{\log\gamma_+^{(1)}-\log\gamma_+^{(0)}}{\Gamma_1-\Gamma_0},
\label{eq:tplus-star}
\end{equation}
and it is a local minimum of $h_+$.  On each connected portion of a sign branch,
the branch maximum is obtained by comparing the relevant stationary value, if it
belongs to the portion, with endpoint and asymptotic values.
\end{fact}

This fact follows upon inspection of the cases below.
If $\alpha\beta\ge0$, then $\abs{\alpha}+\abs{\beta}=\abs{\alpha+\beta}$.
Using Eq.~\eqref{eq:parameters}, one obtains Eq.~\eqref{eq:hminus}.  Moreover,
\begin{equation}
 \dot h_-(t)=-2\gamma_-^{(0)}e^{-\Gamma_0t}
 +2\gamma_-^{(1)}e^{-\Gamma_1t},
 \label{eq:hminus-derivative}
\end{equation}
so the stationary condition gives Eq.~\eqref{eq:tminus-star}, and
\begin{equation}
 \ddot h_-(t_-^*)=2\gamma_-^{(0)}e^{-\Gamma_0t_-^*}(\Gamma_0-\Gamma_1)<0.
 \label{eq:hminus-second}
\end{equation}
The case $\alpha\beta<0$ is identical, with
$\abs{\alpha}+\abs{\beta}=\abs{\alpha-\beta}$.  Differentiating
Eq.~\eqref{eq:hplus} gives
\begin{equation}
 \dot h_+(t)=2\gamma_+^{(0)}e^{-\Gamma_0t}
 -2\gamma_+^{(1)}e^{-\Gamma_1t},
 \label{eq:hplus-derivative}
\end{equation}
and hence Eq.~\eqref{eq:tplus-star}; the second derivative at $t_+^*$ is
positive.  Since each derivative changes sign only once, it remains only to
compare the stationary value with branch endpoints and, on unbounded branches,
with the limiting value $L$.


\section{Entanglement-assisted discrimination}
\label{sec:assisted}

Using Eq.~\eqref{eq:Dent-def-general}, define
\begin{equation}
 D_{\rm ent}(t):=\max_\Psi\norm{(\id\otimes\Delta_t)(\Psi)}_1
 =\norm{\Delta_t}_\diamond,
 \label{eq:Dent-def}
\end{equation}
where the maximum may be restricted to pure two-qubit states.

\begin{lemma}[Schmidt reduction]
\label{lem:diamond-reduction}
For the phase-covariant map $\Delta_t$ in Eq.~\eqref{eq:matrixunits}, there is
an optimal assisted input in the Schmidt family
\begin{equation}
 \ket{\Psi_x}=\sqrt{x}\ket{0_A0_S}+\sqrt{1-x}\ket{1_A1_S},
 \qquad 0\le x\le1.
 \label{eq:Schmidt-family}
\end{equation}
Consequently,
\begin{equation}
 D_{\rm ent}(t)=\max_{0\le x\le1}D_{\rm ent}(x;t),
 \label{eq:Dent-one-variable}
\end{equation}
where
\begin{equation}
D_{\rm ent}(x;t)=
\frac12\max\{\abs{u_x},w_x\}
+\frac12\left[x\abs{\alpha+\beta}+(1-x)\abs{\alpha-\beta}\right],
\label{eq:Dent-x}
\end{equation}
with
\begin{equation}
 u_x=\beta+\alpha(2x-1),
 \qquad
 w_x=\sqrt{\left[\alpha+\beta(2x-1)\right]^2+16\zeta^2x(1-x)}.
 \label{eq:u-w}
\end{equation}
\end{lemma}

\begin{proof}
The diamond-norm semidefinite program for a Hermiticity-preserving map can be
written in terms of the Choi operator
\begin{equation}
 J(\Delta_t)=(\id\otimes\Delta_t)(\ket{\Omega}\bra{\Omega}),
 \qquad
 \ket{\Omega}=\ket{00}+\ket{11}.
 \label{eq:choi-def}
\end{equation}
Because of the covariance in Eq.~\eqref{eq:phase-covariance}, the Choi operator
is invariant under $\overline U_\theta\otimes U_\theta$. Since the feasible set and the objective of the SDP are invariant under this
compact group action, averaging any optimal feasible point over the group gives
another optimal feasible point with a diagonal system marginal, namely 
\begin{equation}
 \rho=x\ketbra{0}{0}+(1-x)\ketbra{1}{1}.
 \label{eq:diagonal-rho}
\end{equation}
For fixed $\rho$, the SDP value is
\begin{equation}
 \norm{(\rho^{1/2}\otimes I)J(\Delta_t)(\rho^{1/2}\otimes I)}_1,
 \label{eq:fixed-rho-value}
\end{equation}
which is precisely the output norm produced by the purification in
Eq.~\eqref{eq:Schmidt-family}.  This proves the reduction to one parameter.
For this input, in the basis $\basisB$,
\begin{equation}
\Omega_x=(\id\otimes\Delta_t)(\ket{\Psi_x}\bra{\Psi_x})=
\begin{pmatrix}
 {x(\alpha+\beta)\over2} & 0 & 0 & \zeta\sqrt{x(1-x)}\\
 0 & -{x(\alpha+\beta)\over2} & 0 & 0\\
 0 & 0 & {(1-x)(\alpha-\beta)\over2} & 0\\
 \zeta\sqrt{x(1-x)} & 0 & 0 & -{(1-x)(\alpha-\beta)\over2}
\end{pmatrix}.
\label{eq:Omegax-block}
\end{equation}
The two isolated diagonal entries contribute the second term of
Eq.~\eqref{eq:Dent-x}.  The remaining $2\times2$ block has eigenvalues
\begin{equation}
 \omega_\pm(x)=\frac14\left(u_x\pm w_x\right),
 \label{eq:omega-eigenvalues}
\end{equation}
and its trace-norm contribution is $\frac12\max\{\abs{u_x},w_x\}$.
\end{proof}

If both semigroups are genuinely relaxing, namely
$\Gamma_0,\Gamma_1>0$, then Eq.\eqref{eq:Dent-x} yields
\begin{equation}
\lim_{t\to\infty}D_{\rm ent}(t) = |\chi_0-\chi_1| = \lim_{t\to\infty}D_{\rm
sep}(t).
\end{equation}
Hence, if the unassisted strategy attains a value strictly larger than
the common asymptotic value at a finite time, then the assisted
strategy also has a finite-time optimum, because \(D_{\rm ent}(t)\ge
D_{\rm sep}(t)\).  The converse need not hold.

\begin{remark}
\label{rem:schmidt-nonunique}
The Schmidt-form input need not be unique.  The reduction uses the common
$U(1)$ phase covariance of the two channels.  A strict operational advantage at
fixed time requires $D_{\rm ent}(t)>D_{\rm sep}(t)$; after timing optimization it
requires Eq.~\eqref{eq:global-assisted-criterion}.
\end{remark}

The identity
\begin{equation}
 w_x^2-u_x^2=4x(1-x)(4\zeta^2+\alpha^2-\beta^2)
 \label{eq:wminus-u}
\end{equation}
is useful below.  If $4\zeta^2\le \beta^2-
\alpha^2$, then $\abs{u_x}\ge w_x$ for all $x$, and the maximum in
Eq.~\eqref{eq:Dent-x} is attained at $x=0$ or $x=1$.  In that case side
entanglement cannot improve the fixed-time distinguishability.  A nontrivial entangled optimum therefore requires
an interior Schmidt parameter, although an interior optimum alone does not prove
a global timing advantage.

\begin{proposition}[Fixed-time entanglement advantage]
\label{prop:entass}
Fix $t$ and write $\alpha=\alpha(t)$, $\beta=\beta(t)$, and
$\zeta=\zeta(t)$.  Then
\begin{equation}
 D_{\rm ent}(t)>D_{\rm sep}(t)
 \label{eq:fixed-time-advantage}
\end{equation}
whenever one of the following conditions holds:
\begin{enumerate}[(i)]
 \item $\beta=0$ and $\alpha\zeta\ne0$;
 \item $\alpha=0$ and $0<\beta^2<4\zeta^2$;
 \item $\zeta^2\le\beta^2<4\zeta^2-2\abs{\alpha\beta}-\min\{\alpha^2,\beta^2\}$.
\end{enumerate}
Consequently, the assisted strategy is strictly better after timing
optimization if there exists $t_0\ge0$ such that
\begin{equation}
 D_{\rm ent}(t_0)>\sup_{t\ge0}D_{\rm sep}(t).
 \label{eq:sufficient-global-advantage}
\end{equation}
\end{proposition}

\begin{proof}
If $\beta=0$, then $D_{\rm sep}=\sqrt{\alpha^2+\zeta^2}$ and
from Lemma \ref{lem:diamond-reduction} (Eq.~\eqref{eq:Dent-x} at $x=1/2$) we obtain
\begin{equation}
 D_{\rm ent}(t)\ge\frac12\left(\abs{\alpha}+\sqrt{\alpha^2+4\zeta^2}\right)>
 \sqrt{\alpha^2+\zeta^2},
 \label{eq:beta-zero-proof}
\end{equation}
whenever $\alpha\zeta\ne0$.  If $\alpha=0$, then
$D_{\rm sep}=\max\{\abs{\beta},\abs{\zeta}\}$, while from Lemma \ref{lem:diamond-reduction}
\begin{equation}
 D_{\rm ent}(t)\ge D_{\rm ent}(1/2;t)=\abs{\zeta}+\frac{\abs{\beta}}2.
 \label{eq:alpha-zero-proof}
\end{equation}
This is strictly larger than $D_{\rm sep}$ exactly when
$0<\abs{\beta}<2\abs{\zeta}$.

For the third condition, the assumption $\zeta^2\le\beta^2$ puts the
unassisted optimum on the boundary, so
\begin{equation}
D_{\rm sep}(t)=\abs{\alpha}+\abs{\beta}.
\label{eq:prop-third-Dsep-boundary}
\end{equation}
The strict upper bound in condition (iii) implies
\begin{equation}
4\zeta^2+\alpha^2-\beta^2>0,
\label{eq:prop-third-K-positive}
\end{equation}
and therefore, by Eq.~\eqref{eq:wminus-u}, $w_x>\abs{u_x}$ for all
$0<x<1$.  Thus near the relevant endpoint the assisted expression is
\begin{equation}
D_{\rm ent}(x;t)=
\frac12 w_x+\frac12\left[x\abs{\alpha+\beta}+(1-x)\abs{\alpha-\beta}\right].
\label{eq:prop-third-local-Dent}
\end{equation}
If $\alpha\beta<0$, the unassisted boundary value is attained at $x=0$.
A direct one-sided differentiation gives
\begin{equation}
\left.\frac{d}{dx}D_{\rm ent}(x;t)\right|_{x=0^+}
=
-\abs{\beta}+\frac{4\zeta^2}{\abs{\alpha-\beta}}
+\frac12\left(\abs{\alpha+\beta}-\abs{\alpha-\beta}\right).
\label{eq:prop-third-derivative-left}
\end{equation}
Condition (iii) makes the right-hand side positive.  Therefore
$D_{\rm ent}(x;t)>D_{\rm ent}(0;t)=D_{\rm sep}(t)$ for all sufficiently small
positive $x$.  If $\alpha\beta\ge0$, the unassisted boundary value is attained
at $x=1$, and
\begin{equation}
\left.\frac{d}{dx}D_{\rm ent}(x;t)\right|_{x=1^-}
=
\abs{\beta}-\frac{4\zeta^2}{\abs{\alpha+\beta}}
+\frac12\left(\abs{\alpha+\beta}-\abs{\alpha-\beta}\right).
\label{eq:prop-third-derivative-right}
\end{equation}
Condition (iii) makes this one-sided derivative negative.  Hence replacing
$x=1$ by $x=1-\varepsilon$, with $\varepsilon>0$ small, strictly increases the
assisted distinguishability.  In both sign cases an interior Schmidt parameter
gives a value strictly larger than the unassisted boundary value.  The final
assertion follows immediately from Eqs.~\eqref{eq:ultimate-errors} and
\eqref{eq:sufficient-global-advantage}.
\end{proof}

An interior maximizer \(0<x<1\) certifies that an entangled input is
optimal at that fixed time, but a global advantage requires comparison
with the best unassisted value over all times.

\begin{remark}
\label{rem:branch-structure}
The active branches in Eqs.~\eqref{eq:Dsep} and \eqref{eq:Dent-x} need not be
separated by a single transition time.  The branch boundaries are zero sets of
exponential polynomials, such as (looking at Eq.~\eqref{eq:Dsep})
\begin{equation}
Q_{\rm sep}(t)=\zeta(t)^2-\beta(t)^2-\abs{\alpha(t)\beta(t)}
\label{eq:Qsep-branch-test}
\end{equation}
for the unassisted interior branch and (looking at Eq.~\eqref{eq:wminus-u})
\begin{equation}
K(t)=4\zeta(t)^2+\alpha(t)^2-\beta(t)^2
\label{eq:K-assisted-branch-test}
\end{equation}
for the assisted $w_x$-dominated regime.  Without additional monotonicity
assumptions these functions can have several positive zeros.  Thus the branch
domains should be understood, in general, as unions of connected intervals.  In
numerical optimizations this causes no difficulty: one may either optimize the
original one-variable expressions directly or compare the maxima over all
connected branch components.
\end{remark}


\section{Examples}
\label{sec:examples}

The preceding analytical results isolate several distinct mechanisms for
optimal timing and for the usefulness of side entanglement.  The following
examples illustrate the interior unassisted branch, the boundary unassisted
branch, and the three sufficient fixed-time mechanisms in
Proposition~\ref{prop:entass}.  The numerical optimizations use
Eqs.~\eqref{eq:Dsep} and \eqref{eq:Dent-one-variable}; branch changes are handled
as described in Remark~\ref{rem:branch-structure}.

\subsection{Interior unassisted optimum}
\label{sec:ex-interior}
Consider the case of two processes with the same
longitudinal dynamics and different transverse damping:
\begin{equation}
 (\gamma_+^{(0)},\gamma_-^{(0)},\gamma_z^{(0)})=(1,1,0),
 \qquad
 (\gamma_+^{(1)},\gamma_-^{(1)},\gamma_z^{(1)})=(1,1,1/2).
 \label{eq:lemma1-example-rates}
\end{equation}
Then
\begin{equation}
 \alpha(t)=0,
 \qquad
 \beta(t)=0,
 \qquad
 \zeta(t)=e^{-t}-e^{-2t}.
 \label{eq:lemma1-example-abz}
\end{equation}
Thus
\begin{equation}
D_{\rm sep}(t)=e^{-t}-e^{-2t},
\qquad
 t^*=\log2,
 \qquad
 D_{\rm sep}(t^*)=\frac14.
 \label{eq:lemma1-example-result}
\end{equation}
In this purely transverse case side entanglement does not improve the
trace-norm distinguishability.  Indeed, Eq.~\eqref{eq:Dent-x} gives
\begin{equation}
D_{\rm ent}(t)=D_{\rm sep}(t)=|\zeta(t)|
\label{eq:lemma1-ent-useless}
\end{equation}
for every $t$.  See Fig.~\ref{fig:lemma1-interior}.

\begin{figure}[t] \centering
\includegraphics[width=0.52\textwidth]{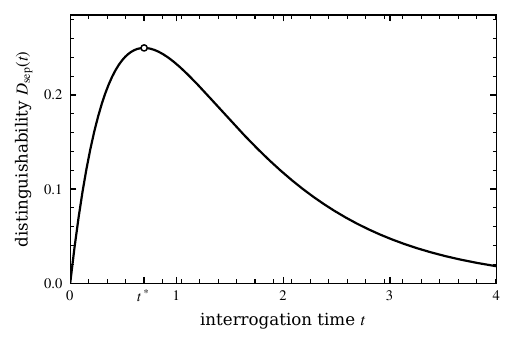}
\caption{Interior-branch timing for the rates in
Eq.~\eqref{eq:lemma1-example-rates}.  The unassisted
distinguishability reaches its maximum $D_{\rm sep}^{*}=1/4$ at
$t^*=\log2$. Side entanglement gives no improvement in this example.}
\label{fig:lemma1-interior} \end{figure}

\subsection{Boundary optimum without entanglement}
\label{sec:ex-boundary}

Take
\begin{equation}
(\gamma_+^{(0)},\gamma_-^{(0)},\gamma_z^{(0)})=(0.2,0.8,0.5), \qquad
(\gamma_+^{(1)},\gamma_-^{(1)},\gamma_z^{(1)})=(0.5,1.5,0.25).
\label{eq:lemma2-example-rates}
\end{equation}
Here the transverse
contractions coincide, namely $\zeta(t)=0$. One has $\Gamma_0=1$ and
$\Gamma_1=2$, and Fact~\ref{cond2} gives
\begin{equation}
t_-^*=\log\frac{1.5}{0.8}\simeq0.6286, \qquad D_{\rm
sep}(t_-^*)=D_{\rm ent}(t_-^*)\simeq0.3267.
\label{eq:lemma2-example-result}
\end{equation}
The asymptotic value
is \(\abs{\chi_0-\chi_1}=0.1\), while Fact~\ref{cond2} gives a finite
boundary-branch maximum at \(t_-^*\simeq0.6286\). Since $\zeta=0$,
entanglement is useless. See Fig.~\ref{fig:lemma2-boundary}.

\begin{figure}[t]
\centering \includegraphics[width=0.52\textwidth]{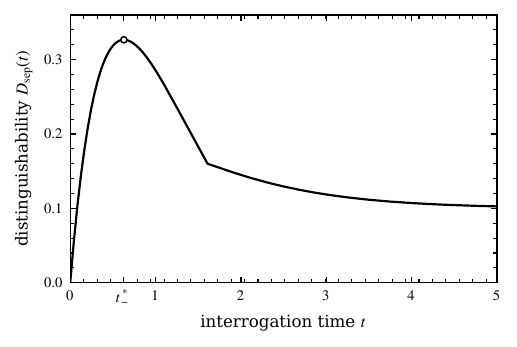}
\caption{Boundary-branch timing for the rates in
  Eq.~\eqref{eq:lemma2-example-rates}.  The finite maximum occurs at
  $t_-^*\simeq0.6286$, and side entanglement gives no improvement
  because $\zeta(t)=0$.}
\label{fig:lemma2-boundary}
\end{figure}

\subsection{Finite-time advantage with $\beta=0$}
\label{sec:ex-beta-zero}

Let
\begin{equation}
(\gamma_+^{(0)},\gamma_-^{(0)},\gamma_z^{(0)})=(0.7,0.3,0),
\qquad
(\gamma_+^{(1)},\gamma_-^{(1)},\gamma_z^{(1)})=(0.3,0.7,1).
\label{eq:finite-ent-example-rates}
\end{equation}
Both dynamical processes have $\Gamma=1$, while their stationary polarizations and
transverse damping rates differ. With $u=e^{-t}$,
\begin{equation}
\alpha(t)=0.8(1-u),
\qquad
\beta(t)=0,
\qquad
\zeta(t)=u^{1/2}(1-u^2).
\label{eq:beta-zero-example-abz}
\end{equation}
The conditions of Proposition \ref{prop:entass} (i) are fulfilled.
The unassisted supremum is
\begin{equation}
\sup_{t\ge0}D_{\rm sep}(t)=0.8,
\qquad
p_E^*=0.3,
\label{eq:beta-zero-sep-result}
\end{equation}
and it is reached only as $t\to\infty$.  The assisted optimum is
reached at finite time:
\begin{equation}
 t^*\simeq1.4508,
 \qquad
 x^*=1/2,
\qquad D_{\rm ent}(t^*)\simeq0.8568,
\qquad
\widetilde p_E^{\,*}\simeq0.2858.
\label{eq:beta-zero-ent-result}
\end{equation}
See Fig.~\ref{fig:beta-zero-finite-vs-infinite}. 

\begin{figure}[t]
\centering
\includegraphics[width=0.52\textwidth]{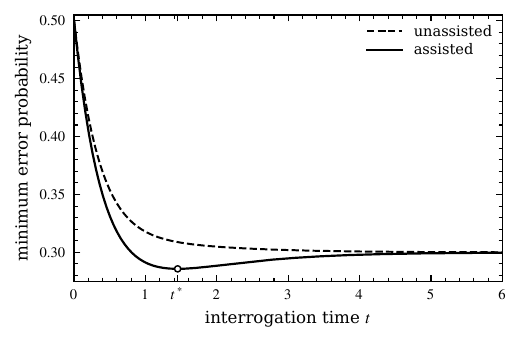}
\caption{Minimum error probability versus interrogation time for the
  rates in Eq.~\eqref{eq:finite-ent-example-rates}.  The unassisted
  infimum $p_E^*=0.3$ is reached only asymptotically, whereas the
  entanglement-assisted strategy reaches a smaller error probability
  $\widetilde p_E^{\,*}\simeq0.2858$ at the finite time
  $t^*\simeq1.4508$.}
\label{fig:beta-zero-finite-vs-infinite}
\end{figure}

\subsection{Threshold behaviour with $\alpha=0$}
\label{sec:ex-alpha-zero}

Compare the symmetric damping process $(1,1,0)$ with the pure-dephasing process
$(0,0,r)$.  Setting $u=e^{-2t}$ gives
\begin{equation}
\alpha(t)=0,
\qquad
\beta(t)=u-1,
\qquad
\zeta(t)=u^{1/2}-u^r.
\label{eq:threshold-abz}
\end{equation}
Therefore
\begin{equation}
D_{\rm sep}(t)=\max\{1-u,\abs{u^{1/2}-u^r}\},
\qquad
\sup_{t\ge0}D_{\rm sep}(t)=1,
\label{eq:threshold-Dsep}
\end{equation}
with the supremum reached only asymptotically as $t\to\infty$.  
The assisted strategy improves on the asymptotic unassisted optimum iff
\begin{equation}
\max_{0<u<1}\left[\frac{1-u}{2}+u^r-u^{1/2}\right]>1.
\label{eq:threshold-condition}
\end{equation}
The threshold is
\begin{equation}
 r_c\simeq0.1100.
 \label{eq:threshold-rc}
\end{equation}
For $r<r_c$, the conditions of Proposition \ref{prop:entass} (ii) are fulfilled and
the entanglement-assisted strategy reaches
$\widetilde p_E^{\,*}<1/4$ at finite time, whereas for $r\ge r_c$ the ultimate
assisted and unassisted error probabilities coincide.  For example, at $r=0.1$
one finds
\begin{equation}
 t^*\simeq2.1496,
 \qquad D_{\rm ent}(t^*) \simeq1.0272,
\qquad
\widetilde p_E^{\,*}\simeq0.2432.
\label{eq:threshold-numerics}
\end{equation}
See Fig.~\ref{fig:alpha-zero-threshold}.

\begin{figure}[t]
\centering
\includegraphics[width=0.52\textwidth]{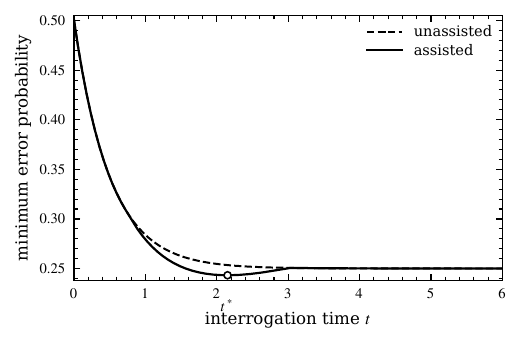}
\caption{Threshold example for discriminating the damping process
  $(1,1,0)$ from the pure-dephasing process $(0,0,r)$, with
  representative value $r=0.1<r_c$. The unassisted infimum is
  $p_E^*=1/4$, while the assisted strategy reaches the smaller
  value $\widetilde p_E^{\,*}\simeq0.2432$ at finite time.}
\label{fig:alpha-zero-threshold}
\end{figure}

\subsection{Advantage from Proposition~\ref{prop:entass}(iii)}
\label{sec:ex-prop-third}
The previous two examples illustrate the special cases $\beta=0$ and
$\alpha=0$.  We now give an example in which all three quantities
$\alpha$, $\beta$, and $\zeta$ are nonzero and the third sufficient
condition of Proposition~\ref{prop:entass} applies.  Let
\begin{equation}
(\gamma_+^{(0)},\gamma_-^{(0)},\gamma_z^{(0)})=(0.6,0.1,0.1),
\qquad
(\gamma_+^{(1)},\gamma_-^{(1)},\gamma_z^{(1)})=(1.3,0.6,0.1).
\label{eq:prop-third-example-rates}
\end{equation}
At $t_0\simeq1.1841$,
\begin{equation}
\alpha(t_0)\simeq0.0729,
\qquad
\beta(t_0)\simeq0.3311,
\qquad
\zeta(t_0)\simeq0.2652,
\label{eq:prop-third-abz}
\end{equation}
and
\begin{equation}
\zeta(t_0)^2\simeq0.0703
\le
\beta(t_0)^2\simeq0.1096
<0.2277
\simeq4\zeta(t_0)^2-2\abs{\alpha(t_0)\beta(t_0)}-
\min\{\alpha(t_0)^2,\beta(t_0)^2\}.
\label{eq:prop-third-condition}
\end{equation}
Hence Proposition~\ref{prop:entass}(iii) applies.  Numerically,
\begin{equation}
D_{\rm sep}(t_0)\simeq0.4040,
\qquad
D_{\rm ent}(t_0)\simeq0.4441,
\qquad
x^*\simeq0.6829.
\label{eq:prop-third-values}
\end{equation}
Since $\sup_tD_{\rm sep}(t)\simeq0.4093$, this fixed-time assisted value also
proves a global timing advantage. See
Fig.~\ref{fig:prop-third-advantage}. 

\begin{figure}[t]
\centering
\includegraphics[width=0.52\textwidth]{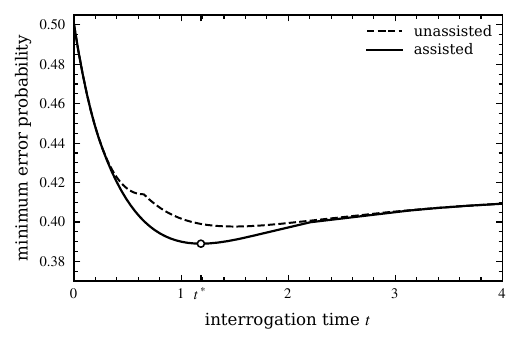}
\caption{Entanglement advantage from Proposition~\ref{prop:entass}(iii) for the rates in Eq.~\eqref{eq:prop-third-example-rates}.  The assisted value at finite time is larger than the globally optimized unassisted value.}
\label{fig:prop-third-advantage}
\end{figure}

\subsection{Generic numerical examples}
\label{sec:generic-numerical-examples}
We finally report two generic numerical instances in which no
simplifying condition on $\alpha$, $\beta$, or $\zeta$ is imposed.
These examples serve as a check that the one-dimensional formula
Eq.~\eqref{eq:Dent-one-variable} captures entanglement advantage also
away from the transparent regimes discussed above.  For
\begin{equation}
 (\gamma_+^{(0)},\gamma_-^{(0)},\gamma_z^{(0)})=(1.1,1.2,0.2),
 \qquad
 (\gamma_+^{(1)},\gamma_-^{(1)},\gamma_z^{(1)})=(1.2,0.95,0.5),
 \label{eq:generic-example-1-rates}
\end{equation}
one finds
\begin{equation}
 \max_tD_{\rm sep}(t)\simeq0.1685 \quad (t\simeq0.9513),
 \qquad
 \max_tD_{\rm ent}(t)\simeq0.1867 \quad (t\simeq0.8069,\ x\simeq0.5960).
 \label{eq:generic-example-1-values}
\end{equation}
For
\begin{equation}
 (\gamma_+^{(0)},\gamma_-^{(0)},\gamma_z^{(0)})=(1.2,0.7,0.5),
 \qquad
 (\gamma_+^{(1)},\gamma_-^{(1)},\gamma_z^{(1)})=(0.9,0.6,0.2),
 \label{eq:generic-example-2-rates}
\end{equation}
one obtains
\begin{equation}
 \max_tD_{\rm sep}(t)\simeq0.2045 \quad (t\simeq0.6857),
 \qquad
 \max_tD_{\rm ent}(t)\simeq0.2417 \quad (t\simeq0.6575,\ x\simeq0.3954).
 \label{eq:generic-example-2-values}
\end{equation}
See Figs. \ref{fig:generic-example-1} and \ref{fig:generic-example-2}.

\begin{figure}[t]
\centering
\includegraphics[width=0.52\textwidth]{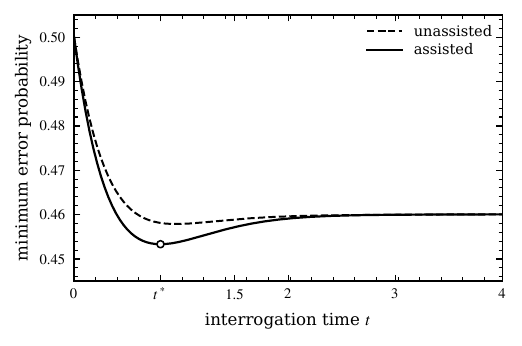}
\caption{Generic numerical example for the rates in
  Eq.~\eqref{eq:generic-example-1-rates}.
  The assisted strategy gives a lower minimum error probability than the unassisted strategy.}
\label{fig:generic-example-1}
\end{figure}

\begin{figure}[t]
\centering
\includegraphics[width=0.52\textwidth]{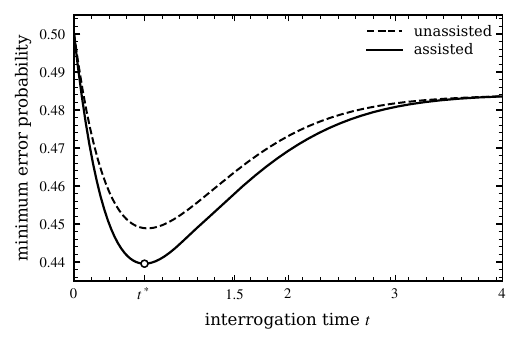}
\caption{Generic numerical example for the rates in
  Eq.~\eqref{eq:generic-example-2-rates}. The assisted strategy gives a lower minimum error probability than the unassisted strategy.}
\label{fig:generic-example-2}
\end{figure}

Before concluding, we record one final remark.
\begin{remark}[Asymptotic assisted advantage]
An assisted advantage need not imply that the assisted optimum is reached at
finite time.  Consider
\[
(\gamma_+^{(0)},\gamma_-^{(0)},\gamma_z^{(0)})=(0,0,0),
\qquad
(\gamma_+^{(1)},\gamma_-^{(1)},\gamma_z^{(1)})=(2,1,0).
\]
Then
\[
\alpha(t)=-\frac13(1-e^{-3t}),\qquad
\beta(t)=1-e^{-3t},\qquad
\zeta(t)=1-e^{-3t/2}.
\]
The unassisted optimum is on the boundary branch for all finite \(t>0\), and
\[
D_{\rm sep}(t)=\frac43(1-e^{-3t}),
\qquad
\sup_tD_{\rm sep}(t)=\frac43,
\]
with the supremum reached only as \(t\to\infty\).  On the other hand, at
\(t=\infty\),
\[
(\alpha,\beta,\zeta)=\left(-\frac13,1,1\right).
\]
Since in this example both \(1-e^{-3t}\) and \(1-e^{-3t/2}\)
increase monotonically to one, the fixed-\(x\) expression in
Eq.~\eqref{eq:Dent-x} is bounded above by its asymptotic value.
Hence the assisted supremum is reached only as \(t\to\infty\),
and Eq.~\eqref{eq:Dent-x} gives
\[
\sup_tD_{\rm ent}(t)=D_{\rm ent}(\infty)=\frac{14}{9},
\]
attained asymptotically with \(x^*=1/3\).  Hence both strategies are
asymptotic, but the entanglement-assisted strategy is strictly better:
\[
\widetilde p_E^*=\frac19<\frac16=p_E^*.
\]
\end{remark}

Within the present Markovian phase-covariant amplitude-damping
semigroups, a strict entanglement advantage at the asymptotic time
requires a surviving transverse component of the difference map.
Since $\lambda(\infty)\ne0 $ only when $\Gamma=\gamma_z=0$, this can
occur only if one of the two candidate semigroups is the identity
semigroup.  If neither semigroup is the identity, then
\(\zeta(\infty)=0\), and the asymptotic difference is diagonal in the
energy basis; consequently $D_{\rm ent}(\infty)=D_{\rm
sep}(\infty)$. Thus the genuinely nontrivial asymptotic entanglement
advantage is specific to comparisons in which one process retains
coherence indefinitely.

\section{Conclusions}
\label{sec:conclusions}
We have investigated minimum-error discrimination of two phase-covariant
amplitude-damping semigroups when both the probe state and the interrogation
time are optimized.  The problem differs qualitatively from static channel
discrimination because the best time at which to compare the two channels is not
known a priori.  It also differs from the Pauli-semigroup setting of
Ref.~\cite{Sacchi26}: the channels considered here are generally non-unital, and
the displacement of the Bloch ball makes the structure of the optimal probes
more delicate.

For unassisted probes we obtained a closed expression for the optimal
trace-norm distinguishability at fixed time.  The result separates the problem
into an interior branch, where the optimal input is a coherent superposition,
and a boundary branch, where an eigenstate of $\sigma_z$ is optimal.

For entanglement-assisted probes we used the common phase covariance of the two
dynamics to reduce the diamond-norm optimization to a one-parameter Schmidt
family.  The resulting expression identifies when a nontrivial Schmidt
parameter can improve the fixed-time distinguishability, and it separates
fixed-time entanglement advantage from advantage after optimization over time.
This distinction is essential: a fixed-time improvement does not automatically
imply that the globally optimized assisted error probability is smaller than the
globally optimized unassisted one.

The examples show several regimes.  In purely transverse discrimination,
entanglement is useless because an unassisted equatorial state already extracts
the full available coherence contrast.  In boundary-branch examples with
$\zeta=0$, side entanglement is again useless.  By contrast, when longitudinal
translation and transverse damping compete, entanglement can give a strict
advantage.  In particular, we exhibited cases in which the unassisted optimum is
approached only as $t\to\infty$, whereas the assisted strategy reaches a lower
error probability at a finite interrogation time.  We also found threshold
behavior and examples in which all three parameters $\alpha$, $\beta$, and
$\zeta$ contribute to the advantage.

These results clarify the role of side entanglement in dynamical channel
discrimination.  Entanglement is not merely a static resource that increases a
norm distance between two channels; in open-system discrimination it can also
change the relevant time scale of the optimal experiment.  This suggests
several directions for further work, including extensions to unequal priors,
finite-time control during the interrogation, non-Markovian dynamics, and
multishot strategies in which sequential or adaptive measurements may exploit
the time dependence of the processes more efficiently.



\begin{thebibliography}{99}
\bibitem{Helstrom76} 
C. W. Helstrom, \emph{Quantum Detection and Estimation Theory} (Academic Press, New York, 1976).
\bibitem{BaeKwek15}
J. Bae and L.-C. Kwek, \textit{Quantum state discrimination and its applications},
J. Phys. A: Math. Theor. \textbf{48}, 083001 (2015).
\bibitem{Sacchi05a} M. F. Sacchi, 
\textit{Optimal discrimination of quantum operations},
Phys. Rev. A \textbf{71}, 062340 (2005).
\bibitem{Watrous18} 
J. Watrous, \emph{The Theory of Quantum Information} (Cambridge University Press, Cambridge, 2018).
\bibitem{Sacchi05c}
M. F. Sacchi, \textit{Entanglement can enhance the distinguishability of entanglement-breaking channels},
Phys. Rev. A \textbf{72}, 014305 (2005).
\bibitem{PianiWatrous09}
M. Piani and J. Watrous, \textit{All entangled states are useful for channel discrimination},
Phys. Rev. Lett. \textbf{102}, 250501 (2009).
\bibitem{Oskouei23}
S. K. Oskouei, S. Mancini, M. Rexiti, \textit{Profitable entanglement for channel discrimination}, 
Proc. R. Soc. A \textbf{479}, 20220796 (2023).
\bibitem{dsk} G. M. D'Ariano, M. F. Sacchi, and J. Kahn,
  \textit{Minimax discrimination of two Pauli channels}, Phys. Rev. A
  \textbf{72}, 052302 (2005).
\bibitem{Wang06} G. Wang, M. Ying, \textit{Unambiguous discrimination
  among quantum operations}, Phys. Rev. A \textbf{73}, 042301 (2006).
\bibitem{Duan09}
R. Duan, Y. Feng, M. Ying, \textit{Perfect distinguishability of quantum operations}, 
Phys. Rev. Lett. \textbf{103}, 210501 (2009).
\bibitem{Chiribella08}
G. Chiribella, G. M. D'Ariano, and P. Perinotti,
\textit{Memory effects in quantum channel discrimination},
Phys. Rev. Lett. \textbf{101}, 180501 (2008).
\bibitem{Harrow10}
A. W. Harrow, A. Hassidim, D. W. Leung, and J. Watrous,
\textit{Adaptive versus nonadaptive strategies for quantum channel discrimination},
Phys. Rev. A \textbf{81}, 032339 (2010).
\bibitem{Gorini76}
V. Gorini, A. Kossakowski, and E. C. G. Sudarshan,
\textit{Completely positive dynamical semigroups of N-level systems},
J. Math. Phys. \textbf{17}, 821 (1976).
\bibitem{Lindblad76}
G. Lindblad,
\textit{On the generators of quantum dynamical semigroups},
Commun. Math. Phys. \textbf{48}, 119 (1976).
\bibitem{BreuerPetruccione02}
H.-P. Breuer and F. Petruccione, \emph{The Theory of Open Quantum Systems}
(Oxford University Press, Oxford, 2002).
\bibitem{RivasHuelga12}
A. Rivas and S. F. Huelga, \emph{Open Quantum Systems: An Introduction}
(Springer, Berlin, 2012).
\bibitem{Sacchi26} M. F. Sacchi, \textit{Entanglement and optimal timing in discriminating quantum dynamical processes}, Phys. Lett. A \textbf{568}, 131199 (2026).
\bibitem{Jevtic15}
S. Jevtic, D. Newman, T. Rudolph, and T. M. Stace,
\textit{Single-qubit thermometry}, Phys. Rev. A \textbf{91}, 012331 (2015).
\bibitem{Candeloro21}
A. Candeloro and M. G. A. Paris,
\textit{Discrimination of Ohmic thermal baths by quantum dephasing probes},
Phys. Rev. A \textbf{103}, 012217 (2021).
\bibitem{Farina19}
D. Farina, V. Cavina, and V. Giovannetti,
\textit{Quantum bath statistics tagging}, Phys. Rev. A \textbf{100}, 042327 (2019).
\bibitem{Farina22}
D. Farina, V. Cavina, M. G. Genoni, and V. Giovannetti,
\textit{Entanglement-assisted, noise-assisted, and monitoring-enhanced quantum bath tagging},
Phys. Rev. A \textbf{106}, 042609 (2022).
\bibitem{Filippov20} 
S. N. Filippov, A. N. Glinov, L. Lepp\"aj\"arvi, 
\textit{Phase covariant qubit dynamics and divisibility},
Lobachevskii J. Math. \textbf{41}, 617 (2020).
\bibitem{Siudzinska22} 
K. Siudzi\'nska, \textit{Phase-covariant mixtures of non-unital qubit maps},
J. Phys. A: Math. Theor. \textbf{55}, 405303 (2022).
\bibitem{Sacchi05b}
M. F. Sacchi, \textit{Minimum error discrimination of Pauli channels}, 
J. Opt. B: Quantum Semiclass. Opt. \textbf{7}, S333 (2005).
\bibitem{Rexiti21}
M. Rexiti, S. Mancini, \textit{Discriminating qubit amplitude damping channels}, 
J. Phys. A: Math. Theor. \textbf{54}, 165303 (2021).
\bibitem{Nakahira21}
K. Nakahira, K. Kato, \textit{Simple Upper and Lower Bounds on the Ultimate Success Probability for Discriminating Arbitrary Finite-Dimensional Quantum Processes}, 
Phys. Rev. Lett. \textbf{126}, 200502 (2021).
\bibitem{Rexiti22}
M. Rexiti, L. Memarzadeh, S. Mancini, \textit{Discrimination of dephasing channels}, 
J. Phys. A: Math. Theor. \textbf{55}, 245301 (2022).
\bibitem{Cooney16}
T. Cooney, M. Mosonyi, and M. M. Wilde,
\textit{Strong converse exponents for a quantum channel discrimination problem and quantum-feedback-assisted communication},
Commun. Math. Phys. \textbf{344}, 797 (2016).
\end{thebibliography}
\end{document}